\documentclass[12pt, letterpaper]{article}
\usepackage[utf8]{inputenc}

\usepackage{hyperref}
\usepackage{amsmath}
\usepackage{amssymb}
\usepackage{amsthm}
\usepackage{stmaryrd}
\usepackage[all]{xy}
\usepackage[margin=2.75cm]{geometry}

\theoremstyle{definition}
\newtheorem{theorem}{Theorem}
\newtheorem{definition}[theorem]{Definition}
\newtheorem{lemma}[theorem]{Lemma}

\theoremstyle{remark}
\newtheorem{remark}[theorem]{Remark}
\newtheorem{example}[theorem]{Example}

\def\eg{{\em e.g.}}
\def\cf{{\em cf.}}

\newcommand{\BP}{\mathbf{P}}
\newcommand{\BC}{\mathbf{C}}
\newcommand{\FOMC}{\text{\sf FOMC}}
\newcommand{\dom}{\text{\rm dom}}

\newcommand{\LQ}{\text{\rm?`}}
\newcommand{\RQ}{\text{\rm?}}

\newcommand{\M}{\mathcal{M}}
\newcommand{\N}{\mathcal{N}}
\renewcommand{\P}{\mathcal{P}}

\renewcommand{\phi}{\varphi}
\renewcommand{\emptyset}{\varnothing}

\title{\texorpdfstring{First-Order Modal $\xi$-Calculus\\On the Aspects of Application and Bisimulation}{First-Order Modal xi-Calculus: On the Aspects of Application and Bisimulation}}
\author{Xinyu Wang\qquad s2010404@jaist.ac.jp}
\date{School of Information Science\\Japan Advanced Institute of Science and Technology\\Asahidai 1--1, Nomi City, Ishikawa Prefecture, Japan}

\begin{document}

\begin{titlepage}
\maketitle\thispagestyle{empty}

\begin{abstract}
\noindent This paper proposes first-order modal $\xi$-calculus as well as genealogical Kripke models. Inspired by modal $\mu$-calculus, first-order modal $\xi$-calculus takes a quite similar form and extends its inductive expressivity onto a different dimension. We elaborate on several vivid examples that demonstrate this logic's profound utility, especially for depicting genealogy of concurrent computer processes. Bisimulation notion for the logic has also been thoroughly examined.
\end{abstract}

\end{titlepage}

\section{Introduction}\label{sec.int}

Modal $\mu$-calculus, first proposed by Kozen in~\cite{Kozen82,Kozen83}, has been well known among modal logicians. Equipped with notations similar to $\lambda$-calculus~\cite{Arnold01}, modal $\mu$-calculus serves as an extremely expressive yet still decidable language (in fact EXPTIME-complete, \cf~\cite{Emerson88,Emerson99}) that is able to cover a wide range of variations of modal logic, especially for those study fields which are tightly related to theoretical computer science, such as temporal logic~\cite{Stirling01}, propositional dynamic logic~\cite{Streett81} as well as finite model theory~\cite{Gradel07}.

Rigidly speaking, semantics of modal $\mu$-calculus is commonly defined via Tarski's fixed-point theorem, \cf~\cite{Bradfield06}. However intuitively, \eg, consider the following modal $\mu$-calculus formula $\nu Z.\phi\land[a]Z$, where $\nu$ is the dual of $\mu$. Here, what this $\nu Z.$ operator actually does is assigning the rest part of the formula, namely $\phi\land[a]Z$, to a formula variable $Z$. Hence, the subformula $\phi\land[a]Z$ is now able to talk about itself via this formula variable $Z$. Through such self-reference, our ``core'' formula $\phi$ gets recursively ``passed down'' along the binary relation $R_a$, so the original formula $\nu Z.\phi\land[a]Z$ intuitively says that $\phi$ is true along every $R_a$-path. So far so good. Then quite naturally, if we regard the binary relation $R_a$ as a \textit{horizontal} dimension in our Kripke model, we may also wonder whether this core formula $\phi$ can similarly be passed down recursively along other types of routes in the model, \eg, on a \textit{vertical} dimension? Actually our answer is yes, and in parallel with modal $\mu$-calculus, we would like to call this novel kind of modal logic as modal $\xi$-calculus.

So after all, where lies this another \textit{vertical} dimension in modal $\xi$-calculus? Let us turn our attention to first-order modal logic, where we are able to discover another natural hierarchy, namely between the whole Kripke model $\M$ itself and any element $a$ that belongs to the (constant) domain of $\M$. We may also note that, \eg, there exists arbitrary unary predicate $P$ so that $Pa$ forms an atomic formula of normal first-order modal logic, but surely, it will be much better if such unary predicates can be \textit{intrinsically} constructed, rather than \textit{extrinsically} designated as extra symbol like $P$. Thus, what if $a$ is a Kripke model as well, exactly the same type as $\M$ so that the hierarchy between $\M$ and $a$ is truly a well-defined \textit{vertical} dimension, and more importantly, any formula can also be evaluated in the model $a$ just as in the model $\M$ so that a formula itself can simply act like an unary operator! Moreover, since now $a$ is the same type of Kripke model like $\M$, then just as $\M$ possesses a domain which contains $a$, of course $a$ also possesses its own domain, and once again, any element in $a$'s domain is also the same type of Kripke model. Through such recursion, a tree-like \textit{vertical} genealogy gets established among a bunch of Kripke models, and so we call this type of Kripke model as genealogical Kripke model. Furthermore, since now any formula acts like an unary predicate, then just as modal $\mu$-calculus, once a formula is equipped with the ability to talk about itself, it can similarly get passed down recursively on the \textit{vertical} dimension, namely along a sequence like $a\in\dom(\M)$, $a'\in\dom(a)$, $\cdots$. Thus more precisely, we call our logic as first-order modal $\xi$-calculus, abbreviated as $\FOMC$. For new notations that we adopt, loosely speaking, if $\phi$ is an $\FOMC$-formula, then adding a pair of question marks around it produces an unary predicate as $\LQ\phi\RQ$, hence $\LQ\phi\RQ a$ is also an $\FOMC$-formula with a structure parallel to that of $Pa$. Also, we introduce the $\xi X.$ operator as $\xi X.\phi$, very similar to modal $\mu$-calculus.

Importantly, here may we provide a few useful tips for readers to follow this paper more smoothly. In fact, no matter how detailedly we manage to explain our na\"{i}ve intuitions behind first-order modal $\xi$-calculus and genealogical Kripke models, instinctive inaccuracy of natural language might still present obstacles toward full comprehension. Therefore, readers are strongly recommended not only to carefully digest all the crucial definitions in Section~\ref{sec.pre}, but also to frequently take some preview on those vivid examples in Section~\ref{sec.exam}, which will certainly help to navigate readers' understanding through intuitive pictures. Generally speaking, although the formal definition may seem a little complex at first, once having completely understood, readers will certainly agree that first-order modal $\xi$-calculus is nothing but just a quite natural mathematical generalization of common modal $\mu$-calculus.

Speaking of practical application, we firstly have to determine how to interpret genealogical Kripke models with suitable realistic meaning. Probably, computer scientists feel it tempting to interpret a model $\M$ as a process on a modern computer, $\M$'s (possibly empty) domain contains all of its children processes, and every child process is again represented by a genealogical Kripke model. This is exactly how a modern computer operating system (such as Linux) commonly handles multiple concurrent processes: on the one hand, there exists a genealogical hierarchy between the parent process and the child process; on the other hand, every process is technically represented by the same data structure and thus applicable to all the same properties, namely in logic, all the same unary predicates. While a lot of other interpretations can also fit into our first-order modal $\xi$-calculus pretty well, within this paper we shall mainly focus on interpreting as computer processes to help to illustrate our logic more vividly. Meanwhile, the other \textit{horizontal} dimension, namely the binary relation $R$ in a genealogical Kripke model, is usually interpreted as temporal logic's chronological future for computer processes, and thus often assumed to be reflexive and transitive~\cite{Gabbay94}. Nevertheless, it may not always be so under every condition, and hence for the most general purpose we choose not to presume any restriction on $R$.

The rest of this paper is organized as the following. Section~\ref{sec.pre} strictly defines mathematical preliminaries. Section~\ref{sec.exam} provides several concrete examples to illustrate powerful practical usage of first-order modal $\xi$-calculus. Section~\ref{sec.bis} studies bisimulation of the logic. Section~\ref{sec.conc} concludes this paper and proposes some feasible directions for future work.

\section{Preliminaries}\label{sec.pre}

This section defines the necessary preliminaries. Should readers find any part of the content difficult to understand, referring to some examples in Section~\ref{sec.exam} as a vivid assistance is strongly suggested.

Also, paragraphs started with a $\bigstar$ mark contain important clarifications. Honestly speaking, this paper indeed includes a handful of concepts and definitions about which readers might be vulnerable to getting confused, perhaps either because the formal definition seems a bit complex (but actually quite natural and intuitive once fully understood), or because it is somewhat different from the ordinary definition (while noticing such subtle difference may prove to be crucial for correct understanding). Thus those clarifications are particularly prepared with very careful elaboration, just in case of various kinds of possible confusion, and so readers are barely encouraged to simply neglect them but had better read those clarifications as one organic, comprehensive part of this entire paper.

\begin{definition}[Formula in Language $\FOMC$]\label{def.for}
Given a countable set of propositional letters $\BP$, and a countable set of constant symbols $\BC$, well formed formula $\phi$ in Language $\FOMC$ is recursively defined as the following BNF:

\begin{align*}
    \phi::=X\mid\top\mid p\mid\LQ\phi\RQ x\mid\LQ\phi\RQ c\mid\neg\phi\mid(\phi\land\phi)\mid\Box\phi\mid\forall x.\phi\mid\xi X.\phi
\end{align*}

where $p\in\BP$, $c\in\BC$, $x$ stands for arbitrary model variable and $X$ stands for arbitrary formula variable. In the rest of this paper, we will follow this routine to use lowercase letters for model variables and uppercase letters for formula variables.

$\perp$, $\lor$, $\to$, $\Diamond$ and $\exists$ are abbreviations defined as common.
\end{definition}

From Definition~\ref{def.for} we can tell that in Language $\FOMC$, there exist two distinct kinds of variables: model variables like $x,y$ and formula variables like $X,Y$. There can also exist some constant symbols like $c,d\in\BC$, which may be called as model constants as well. Thus intuitively, as their names literally suggest, a model variable or a model constant stands for a model (precisely speaking a genealogical Kripke model, introduced later in Definition~\ref{def.model}), while a formula variable represents an $\FOMC$-formula. Readers will gradually become aware of all these notations' respective indispensable function in our logic.

For now let us firstly pay our attention to one issue: usually in logic, variables in a formula may be either free or bound, and so are they here, too. Roughly speaking in an $\FOMC$-formula, a $\forall x.$ operator binds a model variable $x$, while a $\xi X.$ operator (plus some extra condition) binds a formula variable $X$. The formal definition is as the following:

\begin{definition}[Free Variable and Bound Variable]\label{def.bound}
In any given $\FOMC$-formula $\phi$, a model variable $x$ either is bound by the closest $\forall x.$ whose scope includes this $x$; otherwise it is a free model variable, namely, not within any $\forall x.$'s scope.

In any given $\FOMC$-formula $\phi$, a formula variable $X$ either is bound, iff it is firstly within the scope of the closest $\LQ\RQ$ pair, and then bound by the closest $\xi X.$ whose scope includes this $\LQ\RQ$ pair; otherwise it is a free formula variable.
\end{definition}

\begin{example}\label{eg.bound}
For a model variable $x$, whether free or bound is just like ordinary first order logic, \eg, $x$ is bound in $\FOMC$-formula $\forall x.\LQ p\RQ x$ and is free in $\FOMC$-formula $\LQ p\RQ x$.

For a formula variables $X$, however, it must firstly be within the scope of a $\LQ\RQ$ pair, after which could it be bound by some $\xi X.$ whose scope includes that $\LQ\RQ$ pair. This ``extra condition'' requires very careful attention, \eg, $X$ is bound in $\FOMC$-formula $\xi X.\LQ X\RQ c$ but is free in $\FOMC$-formula $\xi X.X$ because it is not within the scope of any $\LQ\RQ$ pair.
\end{example}

\phantom{nothing}

\noindent$\bigstar$ \textbf{Important Clarification} Note that from Definition~\ref{def.bound}, a bound formula variable is not entirely similar to an ordinary bound model variable, \eg, from Example~\ref{eg.bound} we know that $\xi X.X$ does \textit{not} bind its $X$. Intuitively, such stipulation is meant to ensure that circular evaluation will not occur when we later define semantics in Definition~\ref{def.sem}.

Together with the concept of free and bound variables arises the concept of logic sentences. The general idea as such is still quite common: a sentence, also called as a closed formula, is usually defined as a formula without any free variables. However here due to somewhat tangled nature of recursion in $\FOMC$-formulae, only after undergoing very careful analysis can we finally settle down to the following precise conditions in the definition of $\FOMC$-sentences:

\begin{definition}[Sentence in Language $\FOMC$]\label{def.sen}
An $\FOMC$-formula $\phi$ is an $\FOMC$-sentence, iff all of the following conditions hold:

\begin{itemize}
    \item $\phi$ contains neither free model variables nor free formula variables.
    \item If $\xi X.\psi$ is a subformula of $\phi$, then $\xi X.\psi$ contains neither free model variables nor free formula variables.
    \item If $\LQ\psi\RQ x$ or $\LQ\psi\RQ c$ is a subformula of $\phi$, then $\psi$ contains no free model variables.
\end{itemize}
\end{definition}

\noindent$\bigstar$ \textbf{Important Clarification} Note that for any $\FOMC$-formula $\psi$, of course, $\psi$ contains no free model (formula, resp.) variables iff all the model (formula, resp.) variables in $\psi$ are bounded in $\psi$ \textit{per se}. (i.e., it does not mean bounded in $\phi$ even if $\psi$ is a subformula of $\phi$, as in the above Definition~\ref{def.sen}.) Such a narrative convention is indeed natural, and so will be kept throughout this paper without further explicit mentioning.

\begin{remark}\label{rem.sen}
Actually, the formal logical language studied in this paper is $\FOMC$-sentences in Definition~\ref{def.sen}, rather than $\FOMC$-formulae in Definition~\ref{def.for}. Essentially, Definition~\ref{def.for} defines $\FOMC$-formulae by BNF, only as nothing more than an auxiliary notion. Those extra restrictions in Definition~\ref{def.sen} will later on ensure a well-defined semantics (\cf~Theorem~\ref{th.sem}).
\end{remark}

Next, let us turn to the semantical aspect and define genealogical Kripke models, with respect to which semantics of $\FOMC$-sentences will be interpreted:

\begin{definition}[Genealogical Kripke Model]\label{def.model}
A genealogical Kripke model $\M$ is recursively defined as a tuple $(S_\M,R_\M,V_\M,N_\M,I_\M,T_\M)$ where:

\begin{itemize}
    \item $S_\M$ is a non-empty set of possible worlds.
    \item $R_\M\subseteq S_\M\times S_\M$ is a binary relation on $S_\M$.
    \item $V_\M:\BP\to\P(S_\M)$ is a valuation function.
    \item $N_\M$ is a set of genealogical Kripke models.
    \item $I_\M:S_\M\times\BC\to N_\M$ is a partial assignment function.
    \item $T_\M:S_\M\times N_\M\to\bigcup\limits_{\N\in N_\M}S_\N$ is a function such that for any possible world $s\in S_\M$ and any model $\N\in N_\M$, $T_\M(s,\N)\in S_\N$.
\end{itemize}

When the model $\M$ is clear from the context, we can omit $\M$ in the subscript and hence denote the model $\M$ simply as a tuple $(S,R,V,N,I,T)$.
\end{definition}

In case of any misunderstanding about the above Definition~\ref{def.model}, we would like to present (a little longer) clarification as well as several technical remarks right away, while elaboration on the intuitive interpretation of Definition~\ref{def.model} will soon follow after Remark~\ref{rem.model}.

\phantom{nothing}

\noindent$\bigstar$ \textbf{Important Clarification} In Definition~\ref{def.model}, a genealogical Kripke model $\M$ contains some part $N_\M$, which is again a set of genealogical Kripke models. Hence Definition~\ref{def.model} is a recursive definition, just as a logic language is usually defined through BNF. Note that $N_\M$ can be $\emptyset$, which exactly forms the base case of this inductive definition. More specifically speaking, the set-theoretically rigorous version of this definition is transfinitely inductive as follows: at the beginning we define the base class $C_0$, which contains all the genealogical Kripke models $\M$ whose $N_\M$ is $\emptyset$; then for each ordinal $\alpha>0$ we define the inductive class $C_\alpha$, which contains all the genealogical Kripke models $\M$ whose $N_\M$ is a subset of $\bigcup\limits_{\beta<\alpha}C_\beta$; finally the class of all the genealogical Kripke models is $\bigcup\limits_\alpha C_\alpha$, where $\alpha$ ranges over all the ordinals. Therefore by Axiom of Regularity~\cite{Jech03}, circular inclusion will never occur and so we must have $\M\notin N_\M$ for any genealogical Kripke model $\M$; in other words, $N_\M$ is guaranteed to reach the $\emptyset$ dead point after a finite depth of generations, without infinite downward chains.

Just in case of any likely confusion, here we shall elaborate once again. In particular, readers might naturally think it possible for a genealogical Kripke model to have itself as one of its own children models, namely $\M\in N_\M$; or more generally, to contain an infinite downward chain $\cdots\M_3\in N_{\M_2}$, $\M_2\in N_{\M_1}$, $\M_1\in N_\M$. To be honest, we completely agree that these kinds of situations are both conceptually conceivable and mathematically manageable; nonetheless, they simply do \textit{not} appear here because a recursive definition in ZFC automatically rules them out. Such a ``restriction'' is essentially a direct result of the recursive nature of Definition~\ref{def.model}, and thus actually, no extra restriction at all needs to be added into the definition. Readers are suggested to draw an analogy from the recursive BNF in Definition~\ref{def.for}, so that readers can also understand that in a similar way, in Definition~\ref{def.model} we have made neither fault nor negligence. At last, though a formal proof is omitted (which can be easily found in any standard textbook on set theory, \eg,~\cite{Kunen80}), we would like to summarize our above discussion as the following theorem:

\begin{theorem}\label{th.model}
There does not exist an infinite downward chain of genealogical Kripke models $\M_i$, where $i\in\omega$, such that $\forall i\in\omega$, $\M_{i+1}\in N_{\M_i}$.
\end{theorem}

\begin{remark}\label{rem.model}
The following are several other minor points concerning Definition~\ref{def.model}, to which readers are suggested to pay attention as well:

\begin{itemize}
    \item $\N$ and $N$ are different: the former denotes some model, while the latter denotes some part of a model, which is a set of (other) models.
    \item On the one hand, $I$ is generally a \textit{partial} function, i.e., it is not guaranteed to be total. In fact, since the codomain $N$ may be empty, sometimes $I$ is simply impossible to be total.
    \item On the other hand, $T$ is a \textit{total} function, and when its codomain is empty, it is easy to see that its domain must also be empty, so this is not a problem.
\end{itemize}
\end{remark}

Having ensured that mathematically, readers can grasp a correct understanding about Definition~\ref{def.model}, we next provide a brief intuitive explanation on what every part of a genealogical Kripke model stands for. Assume that we interpret the genealogical Kripke model $\M$ as a computer process, then $N_\M$ represents all its children processes, and at every possible state $s\in S_\M$ of the parent process $\M$, for any constant symbol $c\in\BC$, $I_\M(s,c)\in N_\M$ (if defined) selects a specific child process; for any child process $\N\in N_\M$, $T_\M(s,\N)\in S_\N$ designates in which state this child process $\N$ currently is. Therefore fixing a constant symbol $c\in\BC$, its interpretation $I_\M(s,c)$ still depends on $s$ and so may vary from possible world to possible world, thus even though we decide to call $c$ as a model ``constant'', readers should keep aware that $c$ is after all a \textit{local} constant rather than a \textit{global} one. (Of course if needed, we are also able to easily introduce another notion of global model constants.)

Anyway, since as of data structure, both the parent process $\M$ and the child process $\N$ are the same type of object, namely a process, it makes perfect sense that both of them are represented by a genealogical Kripke model so that they share exactly the same mathematical form. (By the way, $\N$ may also have its own children processes as well, which are therefore $\M$'s grandchildren processes.) Here then arises a quite natural question, however: should there be any correlation between the parent process's current state $s\in S_\M$ and the child process's current state $T_\M(s,\N)\in S_\N$, where $\N\in N_\M$? We think the answer depends: in some cases it is quite plausible to assume such kind of relation, while in some other cases both the parent process and the children processes can act independently. Therefore, for the most general purpose, here we choose not to put any extra restrictions upon the relation between the parent process and the children processes.

We now move on to define semantics of $\FOMC$-sentences. To start with, we define interpretations both for any model variable $x$ and for any formula variable $X$:

\begin{definition}[Interpretation]
An interpretation $i$ is a (partial) function. For any $(k,v)\in i$, intuitively, $i$ interprets the key $k$ as the value $v$. We also denote:

\begin{align*}
    i[k:=v]=\left\{\begin{array}{ll}
        i\cup\{(k,v)\}, & \quad\text{if }i(k)\text{ is undefined}\\
        (i\setminus\{(k,v')\})\cup\{(k,v)\}, & \quad\text{if }i(k)=v',\text{ namely }(k,v')\in i
    \end{array}\right.
\end{align*}
\end{definition}

\begin{definition}[Semantics]\label{def.sem}
Given a genealogical Kripke model $\M$ and an $\FOMC$-sentence $\phi$, valuation of $\phi$ over $\M$ (with two interpretations $i$ and $j$, the first one for model variables while the second one for formula variables, and initially $i=j=\emptyset$) is recursively defined as $\llbracket\phi\rrbracket^\M_{\emptyset,\emptyset}\subseteq S_\M$ by the following:

\begin{align*}
    \llbracket X\rrbracket^\M_{i,j}= & \llbracket\psi\rrbracket^\M_{i,j},\text{ where }j(X)=\psi\\
    \llbracket\top\rrbracket^\M_{i,j}= & S_\M\\
    \llbracket p\rrbracket^\M_{i,j}= & V_\M(p)\\
    \llbracket\LQ\phi\RQ x\rrbracket^\M_{i,j}= & \{s\in S_\M\mid T_\M(s,\N)\in\llbracket\phi\rrbracket^\N_{\emptyset,j}\},\text{ where }i(x)=\N\\
    \llbracket\LQ\phi\RQ c\rrbracket^\M_{i,j}= & \{s\in S_\M\mid I_\M(s,c)\text{ is defined and }T_\M(s,I_\M(s,c))\in\llbracket\phi\rrbracket^{I_\M(s,c)}_{\emptyset,j}\}\\
    \llbracket\neg\phi\rrbracket^\M_{i,j}= & S_\M\setminus\llbracket\phi\rrbracket^\M_{i,j}\\
    \llbracket\phi\land\psi\rrbracket^\M_{i,j}= & \llbracket\phi\rrbracket^\M_{i,j}\cap\llbracket\psi\rrbracket^\M_{i,j}\\
    \llbracket\Box\phi\rrbracket^\M_{i,j}= & \{s\in S_\M\mid\text{for all }t\in S_\M\text{ such that }sR_\M t,t\in\llbracket\phi\rrbracket^\M_{i,j}\}\\
    \llbracket\forall x.\phi\rrbracket^\M_{i,j}= & S_\M\cap(\bigcap_{\N\in N_\M}\llbracket\phi\rrbracket^\M_{i[x:=\N],j})\\
    \llbracket\xi X.\phi\rrbracket^\M_{i,j}= & \llbracket\phi\rrbracket^\M_{i,j[X:=\phi]}
\end{align*}

For any genealogical Kripke model $\M$, any possible world $s\in S_\M$ and any $\FOMC$-sentence $\phi$, we denote $\M,s\vDash\phi$ iff $s\in\llbracket\phi\rrbracket^\M_{\emptyset,\emptyset}$.
\end{definition}

Before proving soundness of Definition~\ref{def.sem} in Theorem~\ref{th.sem}, we would like to elaborate further on what on earth this definition of semantics intuitively captures:

\begin{itemize}
    \item For $\llbracket X\rrbracket^\M_{i,j}$, if $j$ interprets the formula variable $X$ as $\psi$, then valuation of $X$ just equals to valuation of $\psi$. Later on we will show that in such a case $j(X)$ is guaranteed to be defined.
    \item For $\llbracket\LQ\phi\RQ x\rrbracket^\M_{i,j}$, the idea is to treat $\LQ\phi\RQ$ as a unary predicate, hence if $i$ interprets the model variable $x$ as some child model $\N\in N_\M$, then for any parent model $\M$'s state $s\in S_\M$, $\LQ\phi\RQ x$ is true at $\M,s$ iff $\phi$ is true at the child model $\N$'s current state $T_\M(s,\N)$, namely $T_\M(s,\N)\in\llbracket\phi\rrbracket^\N_{\emptyset,j}$. Note that the interpretation $j$ for formula variables gets inherited but the interpretation $i$ for model variables is reset to $\emptyset$, because between two models $\M$ and $\N$, their language $\FOMC$ keeps the same but their domains $N_\M$ and $N_\N$ are separate. Later on we will also show that in such a case $i(x)$ is guaranteed to be defined.
    \item For $\llbracket\LQ\phi\RQ c\rrbracket^\M_{i,j}$, we simply attempt to interpret the constant symbol $c$ as a child model $I_\M(s,c)\in N_\M$ and then we can follow the similar idea as above.
    \item For $\llbracket\forall x.\phi\rrbracket^\M_{i,j}$, all possible interpretations for the model variable $x$ as some child model $\N\in N_\M$ gets considered and added into the interpretation $i$, respectively. Then the compound universal formula values as the interaction of all possibilities, also together intersected with $S_\M$ just in case that $N_\M$ may be $\emptyset$.
    \item For $\llbracket\xi X.\phi\rrbracket^\M_{i,j}$, just like modal $\mu$-calculus, this $\xi X.$ operator here intuitively assigns the following formula $\phi$ to the formula variable $X$, thus this interpretation gets added into $j$.
\end{itemize}

\phantom{nothing}

\noindent$\bigstar$ \textbf{Important Clarification} Readers familiar with semantics of modal $\mu$-calculus may wonder, why we have not made use of Tarski's fixed-point theorem in the definition of semantics, just as modal $\mu$-calculus normally does. The reason is that from Theorem~\ref{th.model}, we already know that a genealogical Kripke model can only have a finite depth of generations, hence the fixed point can simply be found by induction through finitely many steps. We also point out that unlike modal $\mu$-calculus, here $\phi$ in $\xi X.\phi$ does not have to be positive, either. Despite these minor technical differences, after all, our first-order modal $\xi$-calculus indeed generalizes the ordinary modal $\mu$-calculus onto the \textit{vertical} dimension through keeping the very core idea totally intact. Readers who may still possess doubts about this point hitherto are extremely eagerly suggested, in any case, to carefully reread Section~\ref{sec.int} --- the very beginning introductory part --- and especially the first several paragraphs, so as to amply convince themselves with quite intuitive justification.

\begin{remark}
Another point we would like to explain about the semantics in Definition~\ref{def.sem} is, why we always designate that $\M,s\nvDash\LQ\phi\RQ c$ when $I_\M(s,c)$ is not defined. This might simply seem like our factitious choice, but actually we have some good reasons to justify it. Let us interpret genealogical Kripke models as computer processes, then $I_\M(s,c)$ is essentially a process pointer that either points to one of $\M$'s children processes $\N\in N_\M$, or is the \textbf{NULL} pointer, which can be regarded as pointing to certain fixed \textbf{NULL} process. Therefore if $I_\M(s,c)$ is undefined, then for any fixed $\FOMC$-formula $\phi$, we should anticipate either always $s\in\llbracket\LQ\phi\RQ c\rrbracket^\M_{i,j}$ or always $s\notin\llbracket\LQ\phi\RQ c\rrbracket^\M_{i,j}$, i.e., depending on how this default \textbf{NULL} process is designed in advance but not depending on the specific current state $s\in S_\M$. Here we just arbitrarily designate that $s\notin\llbracket\LQ\phi\RQ c\rrbracket^\M_{i,j}$ for every $\phi$, which does not really cause an issue, because the expected behavior of $\FOMC$-formula $\LQ\phi\RQ c$ when letting $s\in\llbracket\LQ\phi\RQ c\rrbracket^\M_{i,j}$ can be exactly simulated by $\FOMC$-formula $\neg\LQ\neg\phi\RQ c$ when letting $s\notin\llbracket\LQ\neg\phi\RQ c\rrbracket^\M_{i,j}$.
\end{remark}

At last, before finishing this section of mathematical preliminaries, we have to complete the indispensable task of proving the following Theorem~\ref{th.sem}, showing that our semantics is surely well defined with respect to $\FOMC$-sentences.

\begin{theorem}\label{th.sem}
Semantics in Definition~\ref{def.sem} is well defined.
\end{theorem}

\begin{proof}
For any $\FOMC$-sentence $\phi$, we recursively evaluate $\llbracket\phi\rrbracket^\M_{\emptyset,\emptyset}$.

On the one hand, we claim that this evaluating process always goes on well. Namely, we have to make sure that whenever necessary, $i$ and $j$ will always provide interpretation for any model variable $x$ and any formula variable $X$, respectively. This feature is guaranteed by conditions in Definition~\ref{def.sen}, intuitively as the following:

\begin{itemize}
    \item $\phi$ contains neither free model variables nor free formula variables. This clause is just the ordinary closed-formula condition.
    \item
    \begin{enumerate}
        \item If $\xi X.\psi$ is a subformula of $\phi$, then $\xi X.\psi$ contains no free model variables, or equivalently, $\psi$ contains no free model variables. This half clause is needed because $X$ may appear in a subformula of $\psi$ as $\LQ\cdots X\cdots\RQ x$ or $\LQ\cdots X\cdots\RQ c$, namely, to be passed down to a child model. As explained by the intuition of semantics defined in Definition~\ref{def.sem}, only interpretation for formula variables gets inherited by the child model while interpretation for model variables is reset to $\emptyset$, thus, $\psi$ should not contain any free model variables.
        \item If $\xi X.\psi$ is a subformula of $\phi$, then $\xi X.\psi$ contains no free formula variables, either. This half clause will become of use later on, i.e., in the latter on-the-other-hand part of this proof.
    \end{enumerate}
    \item If $\LQ\psi\RQ x$ or $\LQ\psi\RQ c$ is a subformula of $\phi$, then $\psi$ contains no free model variables, due to the same intuitive reason as the above clause.
\end{itemize}

Based on the above intuitive intentions of all the clauses in Definition~\ref{def.sen}, it then immediately becomes just self-evident that $i(x)$ and $j(X)$ will always be properly defined whenever in need during the whole evaluating process.

On the other hand, we claim that this evaluating process will eventually terminate instead of going on infinitely. In fact, the only possible intermediate step when the current subformula may become more complicated is rewriting $\llbracket X\rrbracket^\N_{i,j}$ to $\llbracket\psi\rrbracket^\N_{i,j}$. This indicates $j(X)=\psi$, so $\xi X.\psi$ is a subformula of $\phi$, and since $\phi$ is an $\FOMC$-sentence, $\xi X.\psi$ contains no free formula variables. Therefore, appearance of any formula variable $Y$ in $\psi$ must be within the scope of a $\LQ\RQ$ pair, and thus if we continue evaluating $\llbracket\psi\rrbracket^\N_{i,j}$ and later on encounter another intermediate stage $\llbracket Y\rrbracket^{\N'}_{i',j'}$, then $\N'$ must be a direct descendant of $\N$. By Theorem~\ref{th.model}, $\N'\neq\N$ and the maximal depth of generations in model $\M$ is finite. Not only do we need finite steps to descend from $\N$ to one of its descendants $\N'$, but also we have finite depth of generations to descend, namely, neither circle nor infinite chain will occur and thus recursively evaluating $\llbracket\phi\rrbracket^\M_{\emptyset,\emptyset}$ is well founded.
\end{proof}

\section{Examples}\label{sec.exam}

We have seen that theoretically, first-order modal $\xi$-calculus is just a very natural mathematical generalization of \textit{horizontal} modal $\mu$-calculus onto the other \textit{vertical} dimension. While practically, first-order modal $\xi$-calculus can depict the recursive structure of computer processes in a genealogical Kripke model, which will be demonstrated by a series of typical examples in the following that cover some of the most interesting topics in concurrency control~\cite{Breshears09}. Starting easily, we first reveal some simple examples without the $\xi X.$ operator.

\phantom{nothing}

\noindent$\bigstar$ \textbf{Important Clarification} As having been introduced in Section~\ref{sec.int}, now the binary relation $R$ is interpreted as temporal logic and thus assumed to be reflexive and transitive~\cite{Goldblatt92}. Therefore for neatness, when drawing the graph for a model we may omit some $R$ relation arrows, but readers should always be aware that the actual relation $R$ must be understood as the reflexive and transitive closure.

\begin{example}\label{exa.de}
As the following graph shows, the parent process $\M$ possesses two children processes $\N_1$ and $\N_2$, with one propositional letter $r$ representing that a process is currently running and one constant symbol $c$ that points to the currently running child process. We have $\M,s_0\vDash\Box\exists x.\LQ r\RQ x\land\neg\exists x.\Box\LQ r\RQ x\land\Box\LQ r\RQ c$, which intuitively reads as: there always exists some running child process, but no single child process is running forever, and it is always true that the currently running child process is currently running. Divergence between the first two clauses is well known in first-order modal logic as \textit{de dicto} vs. \textit{de re}~\cite{Hughes96}, while the last clause simply claims something trivial under our present interpretation.
$$\xymatrix@C=10pt@R=15pt{
\M: & \txt{$s_0$\\$c:=\N_1$}\ar[rr] & & \txt{$s_1$\\$c:=\N_2$}\ar[rr] & & \txt{$s_2$\\$c:=\N_1$}\\
\N_1,running\ar@{--}[ur] & \N_2,stopped\ar@{--}[u] & \N_1,stopped\ar@{--}[ur] & \N_2,running\ar@{--}[u] & \N_1,running\ar@{--}[ur] & \N_2,dead\ar@{--}[u]\\
\N_1: & \txt{$stopped$\\$\neg r$}\ar@<1ex>[r] & \txt{$running$\\$r$}\ar@<1ex>[l]\ar[r] & \txt{$dead$\\$\neg r$}\\
\N_2: & \txt{$stopped$\\$\neg r$}\ar@<1ex>[r] & \txt{$running$\\$r$}\ar@<1ex>[l]\ar[r] & \txt{$dead$\\$\neg r$}
}$$
\end{example}

\begin{example}
As the following graph shows, $\N_1$ and $\N_2$ are two children processes that require both resources $a$ and $b$ to run. $\N_1$ asks for resource $a$ first while $\N_2$ does the opposite, hence in theory we can imagine that a dead lock easily occurs when $\N_1$ occupies resource $a$ and $\N_2$ occupies resource $b$. Nonetheless, in this actual case, $\N_1$ and $\N_2$'s common parent process $\M$ manages to forbid certain combinations of states so as to prevent any possible dead lock. Therefore we have $\M,s_0\vDash\exists x.\exists y.(\LQ\Diamond(a\land\neg b)\RQ x\land\LQ\Diamond(\neg a\land b)\RQ y)\land\neg\Diamond\exists x.\exists y.(\LQ a\land\neg b\RQ x\land\LQ\neg a\land b\RQ y)$, namely, dead lock has been avoided under the parent process $\M$'s supervision.
$$\xymatrix@C=10pt@R=15pt{
& & \N_1,stopped\ar@{--}[dr] & \N_2,waiting\ar@{--}[d] & \N_1,stopped\ar@{--}[dr] & \N_2,running\ar@{--}[d]\\
& & & s_3\ar@(ur,l)[rr] & & s_4\ar@(dl,r)[dllll]\\
\M: & s_0\ar@(dr,l)[rr]\ar@(u,l)[urr] & & s_1\ar@(dr,l)[rr] & & s_2\ar@(ul,ur)[llll]\\
\N_1,stopped\ar@{--}[ur] & \N_2,stopped\ar@{--}[u] & \N_1,waiting\ar@{--}[ur] & \N_2,stopped\ar@{--}[u] & \N_1,running\ar@{--}[ur] & \N_2,stopped\ar@{--}[u]\\
\N_1: & \txt{$stopped$\\$\neg a,\neg b$}\ar[r] & \txt{$waiting$\\$a,\neg b$}\ar[r] & \txt{$running$\\$a,b$}\ar@(ul,ur)[ll]\\
\N_2: & \txt{$stopped$\\$\neg a,\neg b$}\ar[r] & \txt{$waiting$\\$\neg a,b$}\ar[r] & \txt{$running$\\$a,b$}\ar@(ul,ur)[ll]
}$$
\end{example}

\begin{remark}
From the examples above readers may have already noticed that actually, genealogical Kripke models are constructed from an external omniscient viewpoint, and we also implicitly assume that all the states of the child process are completely transparent to the parent process. Therefore, the parent process changes its state in accordance with the child process, even when nothing has really changed in the parent process itself. Nevertheless, $\LQ\RQ$ pairs are still necessary for the parent process to query the states of its children processes, and hence as a possible future research, by putting restrictions on formulae we will then be able to limit the parent process' knowledge about its children processes.
\end{remark}

Next, we turn to a handful of more complicated examples, which clearly demonstrate that in practical applications, the introduction of the $\xi X.$ operator enables the very same core formula to be passed among different processes recursively.

\begin{example}
As the following graph shows, process $\M$ is not running until $s_3$ when its two children processes $\N_1$ and $\N_2$ have both finished, and $\N_1$ again has a child process $\N_3$ whose finish must be waited, too. Due to limit of space here we only draw out one possible history of execution as the root model $\M$. Anyway, we can see that there is always some process running, which can be formulated as $\M,s_0\vDash\Box\xi X.(r\lor\exists x.\LQ X\RQ x)$.
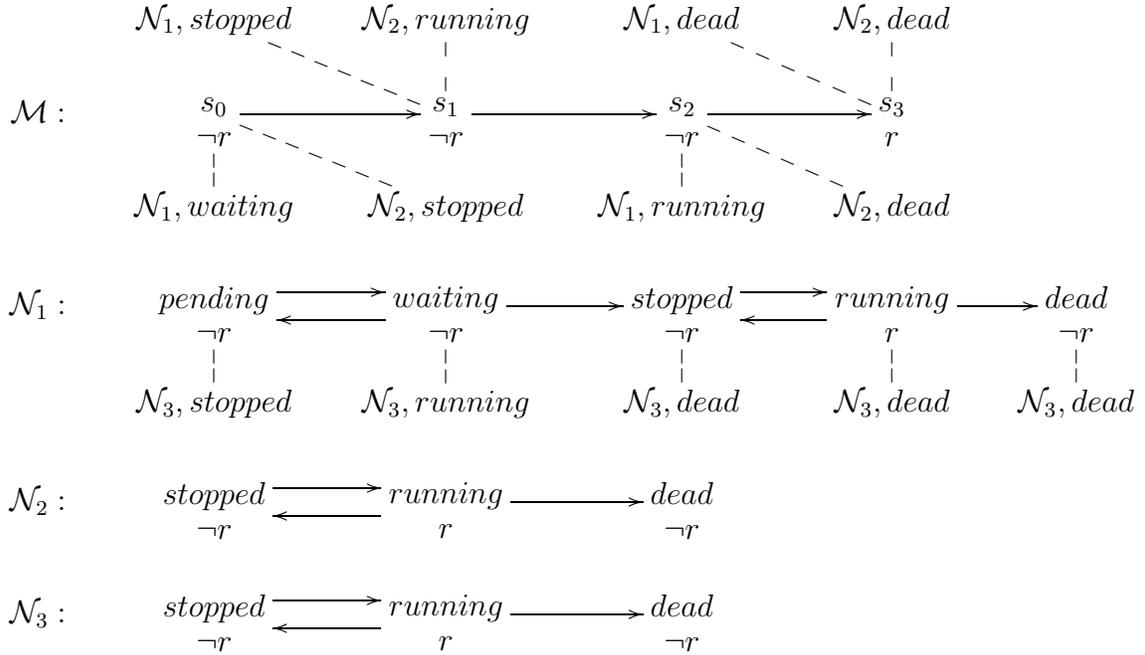
$$\xymatrix@C=20pt@R=15pt{
& \N_1,stopped\ar@{--}[dr] & \N_2,running\ar@{--}[d] & \N_1,dead\ar@{--}[dr] & \N_2,dead\ar@{--}[d]\\
\M: & \txt{$s_0$\\$\neg r$}\ar[r] & \txt{$s_1$\\$\neg r$}\ar[r] & \txt{$s_2$\\$\neg r$}\ar[r] & \txt{$s_3$\\$r$}\\
& \N_1,waiting\ar@{--}[u] & \N_2,stopped\ar@{--}[ul] & \N_1,running\ar@{--}[u] & \N_2,dead\ar@{--}[ul]\\
\N_1: & \txt{$pending$\\$\neg r$}\ar@<1ex>[r] & \txt{$waiting$\\$\neg r$}\ar@<1ex>[l]\ar[r] & \txt{$stopped$\\$\neg r$}\ar@<1ex>[r] & \txt{$running$\\$r$}\ar@<1ex>[l]\ar[r] & \txt{$dead$\\$\neg r$}\\
& \N_3,stopped\ar@{--}[u] & \N_3,running\ar@{--}[u] & \N_3,dead\ar@{--}[u] & \N_3,dead\ar@{--}[u] & \N_3,dead\ar@{--}[u]\\
\N_2: & \txt{$stopped$\\$\neg r$}\ar@<1ex>[r] & \txt{$running$\\$r$}\ar@<1ex>[l]\ar[r] & \txt{$dead$\\$\neg r$}\\
\N_3: & \txt{$stopped$\\$\neg r$}\ar@<1ex>[r] & \txt{$running$\\$r$}\ar@<1ex>[l]\ar[r] & \txt{$dead$\\$\neg r$}
}$$
\end{example}

\begin{example}
We finally present several advanced examples of more complex formulae and briefly explain their intuitive meanings, without providing sample graphs:

\begin{itemize}
    \item $\FOMC$-sentence $\xi X.\forall x.\LQ X\RQ x$ may look like nonsense at first glance, but it is actually valid everywhere. In fact, it simply reflects our assumption that there exists no circular nor infinite genealogical hierarchy, namely, a formal expression of Theorem~\ref{th.model}.
    \item $\FOMC$-sentence $\xi X.(\Box p\land\forall x.\LQ\Diamond X\RQ x)$ says that since the very beginning, the root process keeps to have property $p$, and furthermore as time goes on, property $p$ gradually diffuses down to the children processes layer by layer, so that eventually all the descendant processes will get infected by property $p$. Undoubtedly, such universal kind of spreading pattern --- no matter this property $p$ is a computer virus or anything else --- simply happens everyday and everywhere in our actual computers, network topology, and even human society.
    \item $\FOMC$-sentence $\xi X.(\Diamond\forall y.\LQ\xi Y.(p\land\forall z.\LQ Y\RQ z)\RQ y\land\forall x.\LQ X\RQ x)$ says that for every process, there will be a time when all its descendant processes possess property $p$, \eg, $p$ represents being a dead zombie process so that it can be cleaned up by its parent process~\cite{Tanenbaum15}.
\end{itemize}
\end{example}

\phantom{nothing}

\noindent$\bigstar$ \textbf{Important Clarification} Until now, a variety of pragmatic examples have been vividly exhibited, which should adequately convince readers of first-order modal $\xi$-calculus' robust utility. Very expressive as our logic is, readers may feel dissatisfied about its high complexity and so keep wondering whether the same work can get accomplished by other simpler logics as well. Definitely the answer is no, for expressivity must always be gained at a cost of simplicity, and there is no free lunch. Let us demonstrate this claim through briefly comparing our first-order modal $\xi$-calculus with a number of well-established modal logics.

To begin with, perhaps as one possible alternative to first-order modal $\xi$-calculus, certain sort of multi-modal logic might be suggested, \eg, viewing the instantaneous parent-child relation between states of two processes as a new modality while also introducing another global modality for the $\forall x.$ quantifier. Such kind of multi-modal logic may capture a decent fragment of first-order modal $\xi$-calculus, nonetheless as illustrated in Example~\ref{exa.de}, the crucial \textit{de-dicto-de-re} distinction is known to be uniquely characteristic to first-order modal logic~\cite{Fitting02} and thus cannot get expressed in any multi-modal logic without first-order quantifiers.

Hence next, readers might think about certain transformed type of first-order modal logic, such as the well-known term-modal logic~\cite{Fitting01}. Nevertheless one crucial difference has to be noticed: whereas term-modal logic is based on \textit{epistemic} logic and interprets the Kripke model as a form of knowledge representation, our first-order modal $\xi$-calculus is generally based on \textit{temporal} logic instead and so the binary relation $R$ in a genealogical Kripke model stands for the time order. Such a fundamental divergence effectively makes these two logics totally uncomparable, since they are meant to depict irrelevant phenomena and have their separate uses in reality. In one word as far as we can see, our first-order modal $\xi$-calculus is a quite novel formalization of recursive structures like genealogical Kripke models and thus lies beyond any other existing modal logics.

\section{Bisimulation}\label{sec.bis}

As a quite powerful tool for studying modal logic, the concept of bisimulation plays a particularly significant role over fields related to theoretical computer science, such as process algebra~\cite{Ponse95}. In this section, we propose a bisimulation notion for first-order modal $\xi$-calculus and prove the bisimulation theorem in two directions, namely in one direction bisimulation implies logical equivalence, and in the other direction logical equivalence implies bisimulation under the Hennessy-Milner property~\cite{Blackburn01}.

\begin{definition}[Pointed Genealogical Kripke Model]
A pointed genealogical Kripke model $\M,s$ is a genealogical Kripke model $\M$ with a fixed state $s\in S$.
\end{definition}

The following Definition~\ref{def.bis} provides the notion of bisimulation between two pointed genealogical Kripke models. Before such formal mathematical definition comes, however, we would like to depict an intuitive picture. Now that first-order-like Kripke models are dealt with, say two models $\M$ and $\N$ are bisimilar, then we not only have to designate which world in $\M$ is bisimilar to which world in $\N$ through a binary relation $Z\subseteq S_\M\times S_\N$, but also need to prescribe the correspondence between children models in $\M$'s and $\N$'s domains, viz. $N_\M$ and $N_\N$. Thus suppose $(u,v)\in Z$, then for this pair of counterparts, we need to know \textit{at present} which child model in $N_\M$ corresponds to which child model in $N_\N$ through a similar binary relation $f((u,v))\subseteq N_\M\times N_\N$. Nevertheless, there is no reason to prevent the valuation of $f$ from changing according to the pair $(u,v)$, hence generally speaking, with respect to some fixed $Z$, $f$ should then be a function from $Z$ to $\P(N_\M\times N_\N)$. Therefore the detailed definition goes as the following:

\begin{definition}[Bisimilar Pointed Models]\label{def.bis}
Two pointed models $\M,s$ and $\N,t$ (with the same fixed Language $\FOMC$) are bisimilar, iff there exist a binary relation $Z\subseteq S_\M\times S_\N$ and a function $f:Z\to\P(N_\M\times N_\N)$ such that:

\begin{itemize}
    \item $(s,t)\in Z$.
    \item $\forall(u,v)\in Z\forall p\in\BP$, $u\in V_\M(p)$ iff $v\in V_\N(p)$.
    \item $\forall(u,v)\in Z\forall\M'\in N_\M$, $\exists\N'\in N_\N$ such that $(\M',\N')\in f((u,v))$.
    \item $\forall(u,v)\in Z\forall\N'\in N_\N$, $\exists\M'\in N_\M$ such that $(\M',\N')\in f((u,v))$.
    \item $\forall(u,v)\in Z\forall(\M',\N')\in f((u,v))$, $\M',T_\M(u,\M')$ and $\N',T_\N(v,\N')$ are bisimilar.
    \item $\forall(u,v)\in Z\forall c\in\BC$, either one of the following holds:
    \begin{enumerate}
        \item both $I_\M(u,c)$ and $I_\N(v,c)$ are undefined;
        \item both $I_\M(u,c)$ and $I_\N(v,c)$ are defined, and moreover, $I_\M(u,c),T_\M(u,I_\M(u,c))$ and $I_\N(v,c),T_\N(v,I_\N(v,c))$ are bisimilar.
    \end{enumerate}
    \item $\forall(u,v)\in Z\forall u'\in S_\M$ if $uR_\M u'$, then $\exists v'\in S_\N$, $vR_\N v'$, $(u',v')\in Z$, $f((u,v))\subseteq f((u',v'))$.
    \item $\forall(u,v)\in Z\forall v'\in S_\N$ if $vR_\N v'$, then $\exists u'\in S_\M$, $uR_\M u'$, $(u',v')\in Z$, $f((u,v))\subseteq f((u',v'))$.
\end{itemize}
\end{definition}

\begin{remark}
Once again, Definition~\ref{def.bis} is a recursive definition, for in order to check whether $\M,s$ and $\N,t$ are bisimilar, we have to firstly calculate the bisimilar situations between models in $N_\M$ and models in $N_\N$. Nonetheless, just like Theorem~\ref{th.model}, this issue does not really cause any problem here either and thus Definition~\ref{def.bis} is well defined.
\end{remark}

\begin{theorem}\label{th.bis}
If $\M,s$ and $\N,t$ are bisimilar, then for any $\FOMC$-sentence $\phi$, $\M,s\vDash\phi\iff\N,t\vDash\phi$.
\end{theorem}

\begin{proof}
Suppose $\M,s$ and $\N,t$ are bisimilar with certain fixed $Z$ and $f$, so $(s,t)\in Z$. We inductively prove that, for any possible \textit{intermediate} stages $\llbracket\phi\rrbracket^\M_{i,j}$ and $\llbracket\phi\rrbracket^\N_{k,j}$ when evaluating arbitrary two $\FOMC$-sentences with respect to arbitrary two genealogical Kripke models,\footnote{We do not care about what the \textit{initial} two $\FOMC$-sentences are at the beginning of evaluating; nor do we care about what the \textit{initial} two models are or whether they are just $\M$ or $\N$.} if for every $(x,\M')\in i$, there exists $(x,\N')\in k$ such that $(\M',\N')\in f((s,t))$,\footnote{This condition is \textit{not} asymmetric as it looks like, since $i$ and $k$ are partial functions.} then $s\in\llbracket\phi\rrbracket^\M_{i,j}$ iff $t\in\llbracket\phi\rrbracket^\N_{k,j}$.\footnote{This claim is stronger than the original theorem, because suppose $\phi$ is an $\FOMC$-sentence, then we can simply start evaluating $\phi$ from $\M$ and $\N$ as the initial models, and since the initial stages are also intermediate stages but at the initial stages $i=k=\emptyset$, the condition trivially holds so we have $s\in\llbracket\phi\rrbracket^\M_{\emptyset,\emptyset}$ iff $t\in\llbracket\phi\rrbracket^\N_{\emptyset,\emptyset}$, namely $\M,s\vDash\phi$ iff $\N,t\vDash\phi$.}

As a matter of fact, two nested levels of induction are needed in this proof. The outer level of induction deals with models $\M$'s and $\N$'s maximal depths of generations, whose basic case is when $N_\M=N_\N=\emptyset$. Then, the inner level of induction deals with $\FOMC$-formula $\phi$'s structure. Firstly, at the basic step of the outer level of induction, for the inner level of induction, we only concentrate on those new cases beyond normal propositional modal logic:

\begin{itemize}
    \item Suppose $\phi$ is in the form of $X$. Since the interpretation $j$ for formula variables is the same, and also from proof of Theorem~\ref{th.sem} we know that infinite loop of $X$ will not occur, this case then holds by induction hypothesis of the inner level.
    \item Suppose $\phi$ is in the form of $\LQ\psi\RQ x$. This case is impossible, because now $N_\M=N_\N=\emptyset$ and there is no appropriate interpretation for model variable $x$, contradicting Theorem~\ref{th.sem}.
    \item Suppose $\phi$ is in the form of $\LQ\psi\RQ c$. Also because $N_\M=N_\N=\emptyset$, now both $I_\M(u,c)$ and $I_\N(v,c)$ must be undefined and thus $\llbracket\LQ\psi\RQ c\rrbracket^\M_{i,j}=\llbracket\LQ\psi\RQ c\rrbracket^\N_{k,j}=\emptyset$.
    \item Suppose $\phi$ is in the form of $\forall x.\psi$. Also because $N_\M=N_\N=\emptyset$, we trivially have $\llbracket\forall x.\psi\rrbracket^\M_{i,j}=S_\M$ and $\llbracket\forall x.\psi\rrbracket^\N_{k,j}=S_\N$.
    \item Suppose $\phi$ is in the form of $\xi X.\psi$. This case holds by induction hypothesis of the inner level.
\end{itemize}

Next, we move on to the inductive step of the outer level of induction. For the inner level of induction, we again only mention several technically subtle cases that differ from routine proof:

\begin{itemize}
    \item Suppose $\phi$ is in the form of $\LQ\psi\RQ x$, and $(x,\M')\in i$, $(x,\N')\in k$. Since $(s,t)\in Z$ and $(\M',\N')\in f((s,t))$, from the condition in Definition~\ref{def.bis} we know that $\M',T_\M(s,\M')$ and $\N',T_\N(t,\N')$ are bisimilar, then by induction hypothesis of the outer level we know that $T_\M(s,\M')\in\llbracket\psi\rrbracket^{\M'}_{\emptyset,j}$ iff $T_\N(t,\N')\in\llbracket\psi\rrbracket^{\N'}_{\emptyset,j}$, therefore, $s\in\llbracket\phi\rrbracket^\M_{i,j}$ iff $t\in\llbracket\phi\rrbracket^\N_{k,j}$.
    \item Suppose $\phi$ is in the form of $\LQ\psi\RQ c$. Since $(s,t)\in Z$, from the condition in Definition~\ref{def.bis} we know that either both $I_\M(s,c)$ and $I_\N(t,c)$ are undefined, then we have $s\notin\llbracket\LQ\psi\RQ c\rrbracket^\M_{i,j}$ and $t\notin\llbracket\LQ\psi\RQ c\rrbracket^\N_{k,j}$; or both $I_\M(s,c)$ and $I_\N(t,c)$ are defined, $I_\M(s,c),T_\M(s,I_\M(s,c))$ and $I_\N(t,c),T_\N(t,I_\N(t,c))$ are bisimilar, then by induction hypothesis of the outer level we know that $T_\M(s,I_\M(s,c))\in\llbracket\psi\rrbracket^{I_\M(s,c)}_{\emptyset,j}$ iff $T_\N(t,I_\N(t,c))\in\llbracket\psi\rrbracket^{I_\N(t,c)}_{\emptyset,j}$, therefore, $s\in\llbracket\phi\rrbracket^\M_{i,j}$ iff $t\in\llbracket\phi\rrbracket^\N_{k,j}$.
    \item Suppose $\phi$ is in the form of $\Box\psi$. Since $(s,t)\in Z$, from the condition in Definition~\ref{def.bis} we know that for every $u\in S_\M$ such that $sR_\M u$, there exists $v\in S_\N$ such that $tR_\N v$, $(u,v)\in Z$, $f((s,t))\subseteq f((u,v))$, and vice versa. Hence, fix arbitrary such pair of $(u,v)\in Z$, we only have to show that $u\in\llbracket\psi\rrbracket^\M_{i,j}$ iff $v\in\llbracket\psi\rrbracket^\N_{k,j}$, and because $f((s,t))\subseteq f((u,v))$, for every $(x,\M')\in i$ there exists $(x,\N')\in k$ such that $(\M',\N')\in f((u,v))$, so by induction hypothesis of the inner level the above claim holds, therefore, $s\in\llbracket\phi\rrbracket^\M_{i,j}$ iff $t\in\llbracket\phi\rrbracket^\N_{k,j}$.
    \item Suppose $\phi$ is in the form of $\forall x.\psi$. Since $(s,t)\in Z$, from the condition in Definition~\ref{def.bis} we know that for every $\M'\in N_\M$ there exists $\N'\in N_\N$ such that $(\M',\N')\in f((s,t))$, and vice versa. Hence, fix arbitrary such pair of $(\M',\N')\in f((s,t))$, we only have to show that $s\in\llbracket\psi\rrbracket^\M_{i[x:=\M'],j}$ iff $t\in\llbracket\psi\rrbracket^\N_{k[x:=\N'],j}$, and because $(\M',\N')\in f((s,t))$, for every $(x,\M'')\in i[x:=\M']$ there exists $(x,\N'')\in k[x:=\N']$ such that $(\M'',\N'')\in f((s,t))$, so by induction hypothesis of the inner level the above claim holds, therefore, $s\in\llbracket\phi\rrbracket^\M_{i,j}$ iff $t\in\llbracket\phi\rrbracket^\N_{k,j}$.
\end{itemize}
\end{proof}

Similar to ordinary propositional modal logic~\cite{Hennessy85}, the reverse of Theorem~\ref{th.bis} also holds, provided that both $\M$ and $\N$ are image-finite models. To start with, however, we have to define what an image-finite model is, which is still a recursive definition and actually requires finiteness on two dimensions, namely both horizontally along the binary relation $R$ and vertically along the set of children models $N$:

\begin{definition}[Image-Finite Genealogical Kripke Model]
A genealogical Kripke model $\M$ is image-finite, iff all of the following conditions hold:

\begin{itemize}
    \item For every $s\in S$, the set $\{t\in S\mid sRt\}$ is finite.
    \item The set $N$ is finite.
    \item For every $\N\in N$, the model $\N$ is image-finite.
\end{itemize}
\end{definition}

For proof of the reverse of Theorem~\ref{th.bis}, as a preparation let us firstly focus on the following Lemma~\ref{lem.hm}, which constitutes the major difficulty throughout the entire inductive proof. Just like the diagram method in model theory of first order logic~\cite{Marker02}, here we will also expand the language from $\FOMC$ to $\FOMC'$ by adding a corresponding new constant symbol for every element in the domain, i.e., every child model.

Formally, the expansion works as follows. Suppose $\M$ and $\N$ are two image-finite genealogical Kripke models, $s\in S_\M$, $t\in S_\N$, and for any $\FOMC$-sentence $\phi$, $\M,s\vDash\phi\iff\N,t\vDash\phi$. Let two sets of fresh constant symbols be $\BC_\M=\{c_{\M'}\mid\M'\in N_\M\}$ and $\BC_\N=\{c_{\N'}\mid\N'\in N_\N\}$, we expand the language to $\FOMC'$ by enlarging $\BC'=\BC\cup\BC_\M\cup\BC_\N$, but meanwhile restricting that the appearance of any $c\in\BC_\M\cup\BC_\N$ should not be within the scope of any $\xi X.$ or the scope of any $\LQ\RQ$ pair. Thus, $c$ will never get passed down vertically.

Now suppose there arbitrarily exist two fixed functions $g:N_\M\to N_\N$ and $h:N_\N\to N_\M$, then using these two functions, we naturally view $\M$ and $\N$ also as genealogical Kripke models for the expanded language $\FOMC'$ by letting $I_\M'(u,c_{\M'})=\M'$, $I_\M'(u,c_{\N'})=h(\N')$, $I_\N'(v,c_{\M'})=g(\M')$, $I_\N'(v,c_{\N'})=\N'$ for every $u\in S_\M$, $v\in S_\N$, $c_{\M'}\in\BC_\M$, $c_{\N'}\in\BC_\N$.

\begin{lemma}\label{lem.hm}
Suppose $\M$ and $\N$ are two image-finite genealogical Kripke models, $s\in S_\M$, $t\in S_\N$, and for any $\FOMC$-sentence $\phi$, $\M,s\vDash\phi\iff\N,t\vDash\phi$. Do the expansion as above. Then there exist two functions $g:N_\M\to N_\N$ and $h:N_\N\to N_\M$ such that:

\begin{itemize}
    \item For any natural number $n\in\omega$ and any $sR^n_\M u$, there exists $tR^n_\N v$ so that for any $\FOMC'$-sentence $\phi$, $\M,u\vDash\phi\iff\N,v\vDash\phi$.
    \item For any natural number $m\in\omega$ and any $tR^m_\N v$, there exists $sR^m_\M u$ so that for any $\FOMC'$-sentence $\psi$, $\M,u\vDash\psi\iff\N,v\vDash\psi$.
\end{itemize}
\end{lemma}

\begin{proof}
Consider the nontrivial situation when $N_\M\neq\emptyset$ and $N_\N\neq\emptyset$. We prove by contradiction. Suppose not, namely, for every possible functions $g:N_\M\to N_\N$ and $h:N_\N\to N_\M$, either one of the following two cases holds:\footnote{In the following, $\bigwedge\limits_v$ denotes conjunction over all possible $v$, and the same for $u$, $g$ and $h$, all of which only have finitely many possibilities due to image-finiteness. $\phi$ and $\psi$ generally depend on parameters $v$, $u$, $g$ and $h$, but we just omit too many subscripts for neatness.}

\begin{itemize}
    \item There exist some natural number $n\in\omega$ and $sR^n_\M u$ such that for any $tR^n_\N v$, there exists some $\FOMC'$-sentence $\phi$ so that $\M,u\vDash\phi\nLeftrightarrow\N,v\vDash\phi$, and without loss of generality, we may assume $\M,u\vDash\phi$ and $\N,v\nvDash\phi$. Since the possible number of all the different $v$ is finite, by conjunction we have $\M,s\vDash\Diamond^n\bigwedge\limits_v\phi$ as $sR^n_\M u$, but $\N,t\nvDash\Diamond^n\bigwedge\limits_v\phi$ as $tR^n_\N v$ for every $v$.
    \item There exist some natural number $m\in\omega$ and $tR^m_\N v$ such that for any $sR^m_\M u$, there exists some $\FOMC'$-sentence $\psi$ so that $\M,u\vDash\psi\nLeftrightarrow\N,v\vDash\psi$, and without loss of generality, we may assume $\M,u\nvDash\psi$ and $\N,v\vDash\psi$. Since the possible number of all the different $u$ is finite, by conjunction we have $\N,t\vDash\Diamond^m\bigwedge\limits_u\psi$ as $tR^m_\N v$, but $\M,s\nvDash\Diamond^m\bigwedge\limits_u\psi$ as $sR^m_\M u$ for every $u$, namely, $\M,s\vDash\neg\Diamond^m\bigwedge\limits_u\psi$ but $\N,t\nvDash\neg\Diamond^m\bigwedge\limits_u\psi$.
\end{itemize}

To sum up, as $\M$ is not affected by $g$, we can actually conclude that for every possible functions $g$ and $h$, we always have $\M,s\vDash\bigvee\limits_h(\bigwedge\limits_g\Diamond^n\bigwedge\limits_v\phi\land\bigwedge\limits_g\neg\Diamond^m\bigwedge\limits_u\psi)$. This is an $\FOMC'$-sentence, in the front of which we can add existential quantifiers binding every $c_{\M'}\in\BC_\M$ as well as universal quantifiers binding every $c_{\N'}\in\BC_\N$ so as to convert it back to an $\FOMC$-sentence, namely $\M,s\vDash\exists\vec{c}_{\M'}.\forall\vec{c}_{\N'}.\bigvee\limits_h(\bigwedge\limits_g\Diamond^n\bigwedge\limits_v\phi\land\bigwedge\limits_g\neg\Diamond^m\bigwedge\limits_u\psi)$, and since $\M,s\vDash\phi\iff\N,t\vDash\phi$ for any $\FOMC$-sentence $\phi$, we obtain $\N,t\vDash\exists\vec{c}_{\M'}.\forall\vec{c}_{\N'}.\bigvee\limits_h(\bigwedge\limits_g\Diamond^n\bigwedge\limits_v\phi\land\bigwedge\limits_g\neg\Diamond^m\bigwedge\limits_u\psi)$, and so back again to $\FOMC'$-sentence we can conclude that there exists some $g$ such that $\N,t\vDash\bigvee\limits_h(\bigwedge\limits_g\Diamond^n\bigwedge\limits_v\phi\land\bigwedge\limits_g\neg\Diamond^m\bigwedge\limits_u\psi)$. However, because $\N$ is not affected by $h$, for each disjunctive branch here we always have $\N,t\nvDash\bigwedge\limits_g\Diamond^n\bigwedge\limits_v\phi\land\bigwedge\limits_g\neg\Diamond^m\bigwedge\limits_u\psi$ for any $g$, thus a contradiction.
\end{proof}

At last, we can state the reverse of Theorem~\ref{th.bis}, usually called Hennessy-Milner Theorem, as the following:

\begin{theorem}[Hennessy-Milner]\label{th.hm}
If $\M$ and $\N$ are two image-finite genealogical Kripke models, $s\in S_\M$, $t\in S_\N$, and for any $\FOMC$-sentence $\phi$, $\M,s\vDash\phi\iff\N,t\vDash\phi$. Then $\M,s$ and $\N,t$ are bisimilar.
\end{theorem}

\begin{proof}
According to Lemma~\ref{lem.hm}, fix two functions $g:N_\M\to N_\N$ and $h:N_\N\to N_\M$, let $Z$ be the set of all ``zig-and-zag'' $(u,v)$ pairs described in the two clauses of conclusion in Lemma~\ref{lem.hm}, and let $f((u,v))=g\cup h^{-1}$ for every $(u,v)\in Z$. We still make use of inductive proof with respect to models $\M$'s and $\N$'s maximal depths of generations, and so as induction hypothesis, we also know that this theorem already holds between any pair of $\M',I_\M(u,\M')$ and $\N',I_\N(v,\N')$, where $(u,v)\in Z$, $\M'\in N_\M$ and $\N'\in N_\N$. For any $\FOMC$-sentence $\phi$, by Lemma~\ref{lem.hm} we have $\M,u\vDash\LQ\phi\RQ c_{\M'}\iff\N,v\vDash\LQ\phi\RQ c_{\M'}$, namely $\M',I_\M(u,\M')\vDash\phi\iff g(\M'),I_\N(v,g(\M'))\vDash\phi$, hence by induction hypothesis, $\M',I_\M(u,\M')$ and $g(\M'),I_\N(v,g(\M'))$ are bisimilar, and similarly $h(\N'),I_\M(u,h(\N'))$ and $\N',I_\N(v,\N')$ are bisimilar, too. It is then not difficult to check that such $Z$ and $f$ are indeed a bisimulation between $\M,s$ and $\N,t$.
\end{proof}

\section{Conclusions and Future Work}\label{sec.conc}

This paper proposes first-order modal $\xi$-calculus, a logic for expressively describing genealogical Kripke models. With illustration of several vivid examples, we have witnessed how genealogical Kripke models naturally depict hierarchic and concurrent practical phenomena --- particularly like computer processes, as well as how sentences of first-order modal $\xi$-calculus succinctly but exactly capture lots of most interesting properties of the Kripke structure. The intuitive picture is clear to understand, while mathematically speaking, the logic itself is also technically challenging, which can be perceived through our effortful management in order to precisely characterize the logic's expressivity through a bisimulation notion.

Nonetheless, this paper is still no more than a piece of primitive work, based on which we would be glad to suggest some plausible future research directions:

\begin{itemize}
    \item Despite that genealogical Kripke models form a rather novel kind of shape, we can still attempt to capture part of its features through other formal languages, which will then be put onto a thorough comparison in contrast to our first-order modal $\xi$-calculus, investigating overall aspects such as expressivity, succinctness, standard translation as well as characterization theorem. It should also be worthwhile trying to combine first-order modal $\xi$-calculus with modal $\mu$-calculus into an integrated framework, and studying them all together from a united algebraic perspective.
    \item Being called a sort of logical calculus, sound and complete proof systems need to be further established for our first-order modal $\xi$-calculus, \eg, trying a tableau-kind system might be a good point to start~\cite{Agostino99}. Similar to normal modal logic, we may also reasonably anticipate that different logics dwell on different frame classes, or under different additional restrictions over genealogical Kripke models, \eg, it sounds tempting to naturally consider various heritage correlations between the parent and the children. Moreover, different proof systems can also possess different computational properties, including decidability, complexity, model checking and so on, which should all prove quite essential in future study and practice as this logic's utility in computer science has been heavily suggested.
    \item Last but not least, it seems a quite promising and exciting approach to extend the definition of genealogical Kripke models as well as the corresponding semantics, \eg, what if a model $\M$ is allowed to reflexively refer to itself as one of its own children models, namely $\M\in N_\M$? As shown by Theorem~\ref{th.model}, a non-standard-non-recursive definition is necessary, and so tight relevance to non-well-founded set theory~\cite{Aczel88}, as well as process algebra which has developed very fruitful in theoretical computer science~\cite{Sangiorgi11}, could then be reasonably expected. Another plausible extension of genealogical Kripke models might be to loosen the restriction on constant domain~\cite{Fitting98}, since in reality, construction of new processes as well as destruction of dead zombie ones are just taking place constantly.
\end{itemize}

\section*{Acknowledgment}

\noindent The author owes much thank to Satoshi Tojo and Mizuhito Ogawa for helpful advice about choosing appropriate notations and writing more clearly. A handful of anonymous reviewers on previous versions of the manuscript have also provided tremendous valuable suggestions to help the author improve this paper.

\end{document}